\documentclass[a4paper,11pt]{article}

\usepackage{fullpage}
\usepackage{amsthm}
\usepackage{amsmath}
\usepackage{amssymb}
\usepackage[ruled,vlined]{algorithm2e}

\newtheorem{theorem}{Theorem}
\newtheorem{lemma}{Lemma}

\newtheorem{corollary}[theorem]{Corollary}
\theoremstyle{remark}

\newcommand{\problem}[3]%
{\begin{center}\fbox{\parbox[t]{0.9\textwidth}{\begin{tabular}{rl}%
\multicolumn{2}{l}{\sc #1}\\Instance:& \parbox[t]{0.75\textwidth}{#2}\\ %
Problem: & \parbox[t]{0.75\textwidth}{#3}\\ %
\end{tabular}}}\end{center}}

\newcommand{\egal}{{\sc CEO}}
\newcommand{\egalLong}{{\sc Connected Egalitarian Optimum}}
\newcommand{\egalD}{{\sc Discrete-CEO}}
\newcommand{\egalDLong}{{\sc Discrete Connected Egalitarian Optimum}}
\newcommand{\mcsp}{{\sc MC-Segment-Packing}}
\newcommand{\mcspLong}{{\sc Multiple-Choice Segment Packing}}
\newcommand{\util}{{\sc CUO}}
\newcommand{\utilLong}{{\sc Connected Utilitarian Optimum}}
\newcommand{\utilD}{{\sc Discrete-CUO}}
\newcommand{\utilDLong}{{\sc Discrete Connected Utilitarian Optimum}}

\newcommand{\DM}{{\sc 3DM}}

\title{Computing Socially-Efficient Cake Divisions}
\author{Yonatan Aumann, Yair Dombb and Avinatan Hassidim}
\date{}

\begin{document}

\maketitle

\begin{abstract}
We consider a setting in which a single divisible good (``cake'') needs to be divided between $n$ players, each with a possibly different valuation function over pieces of the cake. For this setting, we address the problem of finding divisions that maximize the {\em social welfare},  focusing on divisions where each player needs to get one contiguous piece of the cake.   
We show that for both the utilitarian and the egalitarian social welfare functions it is NP-hard to find the optimal division. For the utilitarian welfare, we provide a constant factor approximation algorithm, and prove that no FPTAS is possible unless P=NP. For egalitarian welfare, we prove that it is NP-hard to approximate the optimum to any factor smaller than 2. For the case where the number of players is small, we provide an FPT (fixed parameter tractable) FPTAS for both the utilitarian and the egalitarian welfare objectives. 
\end{abstract}

\section{Introduction}

Consider a town with a central conference hall, erected by the municipality for the benefit of the townspeople. Different people and organizations wish to use the hall for their events, each for a possibly different duration. Furthermore, each such event may have its preferences and constraints on the times when it can take place, e.g.~only in the evenings, on weekends, prior to some date, etc. 
How should the municipality allocate the hall to the different events? How do we compute the allocation that maximizes the social welfare provided by this common resource?

A natural setting for analyzing the above problem is that of \emph{cake cutting}, where a single divisible good needs to be divided between several players with possibly different preferences regarding the different parts of the good, or ``cake''. 
The cake cutting problem was first introduced in the 1940's by Steinhaus~\cite{Ste49}, where the goal was to give each of the $n$ players ``their due part'', i.e.~a piece worth at least $\frac{1}{n}$ of the entire cake by their own measure. (In the cake cutting literature, this fairness requirement is termed \emph{proportionality}.) Since then, other objectives have also been considered, with the majority of them requiring that the division be ``fair'', under some definition of fairness (e.g.~envy-freeness).

Here, we address the fundamental problem of \emph{maximizing social welfare} in cake cutting. 
Given a shared resource, the valuation functions of the players for this resource, and a social welfare function, the problem is to find an allocation that maximizes the welfare. Maximizing social welfare has been previously considered for dividing a set of discrete indivisible items, each of which must be given in whole to one player. Here, we consider the problem with a \emph{single, continuously divisible good}, and furthermore focus on the case where each player needs to get a \emph{single contiguous piece} of the good. The contiguity requirement is natural in many settings, e.g.~dividing time (as in the example above), spectrum, and real-estate. 


\paragraph{Results.} We show that the problems of maximizing utilitarian and egalitarian welfare are both NP-hard in the strong sense. For egalitarian welfare, we further show that it is hard to approximate the optimum to any factor smaller than $2$.

For utilitarian welfare, we provide a constant-factor approximation algorithm (note that the strong NP-hardness result implies that no FPTAS exists for the problem). Specifically, our algorithm finds a division with utilitarian welfare within $8+o(1)$ of the optimum, in polynomial time. We also show that approximating both the utilitarian and egalitarian welfare is fixed-parameter-tractable with regards to the parameter $n$ (the number of players).

Finally, we consider the case where the contiguity requirement is dropped, i.e.~each player may get a \emph{collection} of intervals. 
For this setting, we show that the situation varies greatly depending on the model of input. When the valuations are given explicitly to the algorithm, and are piecewise constant, then the  problem can be solved in polynomial time. However, if the algorithm has only oracle access to the valuations, then it is impossible to do any better than an $n$-factor approximation, even if the valuations themselves are piecewise uniform. 


\paragraph{Related Work.}
The problem of maximizing egalitarian welfare when allocating a set of indivisible goods has been extensively considered in the last 15 years \cite{Woe98,AAWY98,BS06,CCK09}. The currently known best algorithms are a polynomial-time algorithm achieving an approximation factor of $O(\sqrt{n}\log^3{n})$~\cite{AS07}, and an algorithm obtaining $\tilde{O}(n^\epsilon)$ approximation in time $n^{O(1/\epsilon)}$, for any $\epsilon = \Omega(\frac{\log{\log{n}}}{\log{n}})$~\cite{CCK09}. Hardness of approximation for this problem, however, is proven only for a factor of $2$ or less~\cite{BD05}. 
Better approximation guarantees are known for more restricted settings, e.g.~when valuations are restricted to having only one possible non-zero value for each item~\cite{BS06,Fei08}. 
Envy minimization in this setting has also been considered in~\cite{LMMS04}, which showed hardness results as well as an FPTAS for the case of players with identical preferences. Unlike this body of work, which considers a {\em non-ordered set} of indivisible items, here we consider a single divisible item, and furthermore require that each player obtain a single contiguous piece of this good. 

Cake cutting problems were first introduced in the 1940's~\cite{Ste49}, and were studied in many variants since then. Various algorithms were proposed for the problem, including a number of ``moving knife'' algorithms, which perform an infinite number of valuations by continuously moving a knife over the cake (for some examples, see~\cite{Str80,EP84} and~\cite{BT95}).
In addition to the algorithmic results, there has also been work on existence theorems~\cite{DS61,Str80}, lower bounds for the complexity of such algorithms (\cite{SW03,Str08,Pro09}, to mention just a few), and a number of books on the subject, e.g.~\cite{BT96,RW98}.

The issue of social welfare in cake cutting was first considered in Caragiannis et al.~\cite{CKKK09} which aimed to quantify the degradation in social welfare that may be caused by different fairness requirements; the same question was studied for connected pieces in~\cite{AD10}. Guo and Conitzer~\cite{GC10}, and Han et al.~\cite{HSTZ11} study the utilitarian welfare achievable by truthful mechanisms for dividing a set of divisible goods, a setting very similar to a cake with piecewise-constant valuations and non-connected pieces. Cohler et al.~\cite{CLPP11} study utilitarian welfare maximization under the envy-freeness requirement (with non-connected pieces). Bei et al.~\cite{BCHTY12} consider a similar question, but with connected pieces, and with proportionality replacing envy-freeness. Also related is the work of Zivan~\cite{Ziv11} which suggests a way for increasing utilitarian welfare using trust.


\section{Model and Definitions}\label{sec:prelim}

\paragraph{Valuation Functions.} 
In our model, the cake is represented by the interval $[0,1]$. Each player $i\in[n]$ (where $[n] = \{1,\ldots,n\}$) has a non-atomic (additive) measure $v_i(\cdot)$, mapping each measurable subset of $[0,1]$ to its value according to player $i$. For most of this work, we are only interested in a value of \emph{intervals} in $[0,1]$, and thus simply write $v_i(a,b)$ for the value of the interval between $a$ and $b$. (Note that since $v_i$ is non-atomic, single points have zero value, and we need not worry about the boundary points $a$ and $b$ themselves.)



We also assume, as common in the cake-cutting literature, that the valuations are \emph{normalized}, i.e.~that $v_i(0,1) = 1$ for every player $i$. However, 
our results hold (with small modifications to the algorithms or complexity) for arbitrary valuations as well.

\paragraph{Social Welfare Functions.} 
We consider two prominent social welfare functions, whose aim is to measure how good each division is for the \emph{whole society}. Let $x$ be a division (to be formally defined shortly); we write $u_i(x)$ to express the value player $i$ obtains from the piece she receives in $x$. The \emph{utilitarian welfare} is defined as the sum of utilities, and we denote $u(x) = \sum_{i\in[n]}{u_i(x)}$. The \emph{egalitarian welfare} is defined as the utility of the worst-off player, and we denote $eg(x) = \min_{i\in[n]}{u_i(x)}$.


\paragraph{Connected Divisions.} 
In this work, we focus on divisions in which every player gets a (disjoint) \emph{single interval} of the cake. Formally, a connected division of the cake $[0,1]$ between $n$ players can be defined as a vector $x = (x_1,\ldots,x_{n-1},\pi)\in [0,1]^{n-1}\times S_n$ (where $S_n$ is the set of all the permutations of $[n]$), having $x_1\leq x_2\leq \dots \leq x_{n-1}$. This is interpreted as making $n-1$ cuts in positions $x_1,\ldots,x_{n-1}$, and allocating the $n$ resulting intervals to the players in the order determined by the permutation $\pi$. Note that the space $X$ of all such divisions is compact; in addition, both utilitarian and egalitarian welfare functions are continuous in $X$ (as the players' valuation functions are all non-atomic). Therefore, for each of these welfare functions there exists a division that maximizes the welfare.

Our main problem is thus the following: \emph{given the players' valuations, what is the (connected) division that maximizes welfare?} Since the two welfare functions considered here obtain maxima in the divisions space, the problem is indeed well-defined. 
For the analysis of these problems, it is useful to consider their decision versions, defined as follows.

\problem{\utilLong~(\util)}{A set $\{v_i\}_{i=1}^n$ of non-atomic measures on $[0,1]$, and a bound $B$.}{Does there exist a connected division $x$ having $u(x) \geq B$?}

\problem{\egalLong~(\egal)}{A set $\{v_i\}_{i=1}^n$ of non-atomic measures on $[0,1]$, and a bound $B$.}{Does there exist a connected division $x$ having $eg(x) \geq B$?}

\paragraph{Complexity and Input Models.}
In order to analyze the complexity of our problems, we must first define how the input is represented. In most of the cake cutting literature, the mechanism is not explicitly given the players' valuation functions; instead, it can \emph{query} the players on their valuations (see e.g.~\cite{EP84,RW98,Str08}). Typically, two types of queries are allowed. In the first, a player $i$ is given points $0\leq a\leq b \leq 1$ and is required to return the value $v_i(a,b)$. In the second type of query, a player $i$ is given a point $a\in[0,1]$ and a value $x$ and is required to return a point $b$ such that $v_i(a,b) = x$; we denote this by $v_i^{-1}(a,x)$.\footnote{Note that using only one type of query it is possible to give approximate answers (in polynomial time) to queries of the other type using binary search.}

In contrast, some more recent works (e.g.~\cite{CLPP10,CLPP11,BCHTY12}) consider a model in which the players give complete descriptions of their valuations 
to the mechanism. In this case, it is usually assumed that the functions have some simple structure, so they can be represented succinctly. Specifically, for each player $i$, let $\nu_i:[0,1] \rightarrow [0,\infty)$ be a \emph{value density function}, such that 
\begin{equation*}
v_i(X) = \int_X{\nu_i(x)}dx \;
\end{equation*}
for every measurable subset $X\in[0,1]$.
Following~\cite{CLPP10}, we say that a valuation function $v_i(\cdot)$ is \emph{piecewise-constant} if its value density function $\nu_i(\cdot)$ is a step function, i.e.~if $[0,1]$ can be partitioned into a \emph{finite} number of intervals such that $\nu_i$ is constant on each interval. If, in addition, there is some constant $c_i$ such that 
$\nu_i(\cdot)$ can only attain the values $0$ or $c_i$, we say that $v_i(\cdot)$ is \emph{piecewise-uniform}.\footnote{Note that in this case the constant $c_i$ is uniquely determined by the total fraction of $[0,1]$ in which $\nu_i(x)\neq 0$, since we require that the valuation of the entire cake should be $1$.} Any piecewise-constant valuation function $v_i(\cdot)$ can be therefore represented by a finite set of subintervals of $[0,1]$ together with the value $\nu_i$ attains in each interval.

Our hardness results show that both of the decision problems above are computationally hard, even when the valuation functions are of the simplest type---piecewise-uniform---and are given explicitly to the mechanism.
In contrast, our positive algorithmic results hold also for the more general oracle model. 
The complexity of our algorithms in this case depends on the number of players $n$ and additionally on a precision parameter $\epsilon$.

\paragraph{The Discrete Variants.}
A convenient preprocessing step in our algorithms will be reducing our problems into ones that are purely combinatorial. More precisely, we consider discrete analogues of the problems, where one is additionally given a set of points in $[0,1]$, and is only allowed to make cuts at points from this set (and not anywhere in $[0,1]$). An alternative interpretation is to consider, instead of a continuous cake, a \emph{sequence of indivisible items}; a connected division in this setting gives each player a \emph{consecutive subsequence} of these items. The discrete variants of our problems are defined as follows:

\problem{\utilDLong~(\utilD)}
{A sequence $A = (a_1,\ldots,a_m)$ of items, a set $\{v_i\}_{i=1}^n$ of valuation functions of the form $v_i:A\rightarrow \mathbb{R}^+$, and a bound $B$.}
{Does there exist a connected division $x$ having $u(x) \geq B$?}

\problem{\egalDLong~(\egalD)}
{A sequence $A = (a_1,\ldots,a_m)$ of items, a set $\{v_i\}_{i=1}^n$ of valuation functions of the form $v_i:A\rightarrow \mathbb{R}^+$, and a bound $B$.}
{Does there exist a connected division $x$ having $eg(x)\geq B$?} 

Our hardness results apply to these ``cleaner'' problems as well.
We note that if we drop the contiguity requirement, allowing players to get any disjoint subsets of $A$, maximizing utilitarian welfare becomes trivial (give each item to the player who values it the most). In contrast, maximizing egalitarian welfare (in the discrete setting with non-connected pieces) is known to be a hard problem~\cite{BD05} and has been studied extensively 
(e.g.~\cite{AS07,CCK09}).

\section{Approximation Algorithms}

In this section we present algorithms that return a division that is guaranteed not to be too far from the social optimum. Throughout this section we assume that the algorithms operate in the (more-general) oracle model. We note that if the valuation functions are given explicitly, and are simple enough (in particular, if they are piecewise-constant), the answer to each oracle query can be computed in time polynomial in the input size.

\subsection{The Discretization Procedure}

As we have previously mentioned, it is often useful to reduce the continuous cake into a sequence of discrete items. We now show that this can indeed be done in a time-efficient manner, and without too much harm to the maximum obtainable welfare. Algorithm~\ref{alg:discretize} below receives a cake instance and a parameter $\epsilon$, and produces a set of cut positions that partition the cake into a set of items. At each point, let $a$ be the position of the last (rightmost) cut. Algorithm~\ref{alg:discretize} asks each player $i$ for the leftmost point $b_i$ such that the value $v_i(a,b_i) = \epsilon$;
it then makes a cut at the leftmost of these points, and repeats the process.

\begin{algorithm}\label{alg:discretize}
\DontPrintSemicolon
\KwData{Oracle access to $v_i(\cdot)$ for each player $i\in[n]$, and $\epsilon > 0$.}
\Begin{
$a \longleftarrow 0$ \;
$C \longleftarrow \{ 0 \}$ \;

\While{$\exists i: v_i(a,1) > \epsilon$}{
	\For{$i\in[n]$}{
		$b_i \longleftarrow v_i^{-1}(a,\epsilon)$ \;
	}
	$b \longleftarrow \min_i{b_i}$ \;
	$C \longleftarrow C \cup \{b \}$ \;
	$a \longleftarrow b$ \;
}
$C \longleftarrow C \cup \{ 1\}$ \;
\Return{$C$}
}
\caption{Discretization Procedure}
\end{algorithm}

Note that the set of cuts $C$ produced by the algorithm induces a sequence of items. Specifically, let $0=c_0< c_1 < \cdots  < c_m=1$ be the cut points in $C$; then, for each $1\leq j\leq m$ create an item $a_j$ with value $v_i(a_j)=v_i(c_{j-1},c_j)$ for player $i\in[n]$.

The following lemma establishes that the set $C$ can be computed efficiently, and that we do not lose much utilitarian welfare by restricting our cuts positions to $C$. We note that a similar claim also holds for egalitarian welfare.

\begin{lemma}\label{lem:discrete}
Let $\{v_i(\cdot)\}_{i\in[n]}$ be a cake instance with $n$ players, and consider some precision parameter $\epsilon$. Then:
\begin{enumerate}
\item Algorithm~\ref{alg:discretize} terminates on this instance in time $O(n^2/\epsilon)$.

\item Let $x$ be a division of the original cake; then there exists a division $y$ making cuts only at points in the set $C$ returned by Algorithm~\ref{alg:discretize}, and having $u(y)\geq u(x) - (n-1)\epsilon$.
\end{enumerate}
\end{lemma}

\begin{proof}
For item (1), note that in each iteration of the while-loop, the value $\sum_{i\in[n]}{v_i(a,1)}$ decreases by at least $\epsilon$; since at the beginning of the algorithm $\sum_{i\in[n]}{v_i(a,1)} = n$, there can be at most $O(n/\epsilon)$ iterations. The claim follows by noting that in each iteration we make $2n$ queries to the oracles.

For item (2), let $x$ be a division of the cake. We define the division $y$ by setting $y_j$ to be the leftmost point $c\in C$ having $c\geq x_j$, for each cut $x_j$ in the division $x$ (the order of the allocation in $y$ is similar to that of $x$). Let $k$ be the player getting the leftmost piece in both divisions; clearly $u_k(y)\geq u_k(x)$. Since for any two consecutive cuts $c',c''\in C$ and any player $i$, $v_i(c',c'')\leq \epsilon$, we also have that for all other players $i\neq k$ it holds that $u_i(y)\geq u_i(x) - \epsilon$, and (2) follows immediately.
\end{proof}

\subsection{Approximating the Utilitarian Welfare}

We now present an approximation algorithm for the problem of maximizing utilitarian welfare; the approximation ratio achieved by our algorithm is $8\big( 1+(n-1)\epsilon \big)$, where $\epsilon$ is a precision parameter, and the running time of the algorithm is polynomial in $n$ and in $1/\epsilon$. As a first step, the algorithm uses the discretization procedure (Algorithm~\ref{alg:discretize}) and obtains a set $A$ of $m$ discrete items. 
We now describe how to approximate the optimal utilitarian welfare for this new instance. The algorithm returns a set $\{(s_i,t_i)\}_{i\in[n]}$, where $s_i$ is the beginning index of $i$'s piece, and $t_i$ is its end index. We also use the notation $(s,t)$ to refer to the consecutive sequence of items $\{s,s+1,\ldots,t-1,t\}$; hence, e.g.~$v_i(s,t) = \sum_{j=s}^t{v_i(j)}$.

\begin{algorithm}\label{alg:cuo}
\DontPrintSemicolon
\KwData{For each player $i\in[n]$ a vector of valuations $v_i:[m]\rightarrow \mathbb{R}^+$.}
\Begin{
$\forall{i\in[n]}: s_i\longleftarrow 0\ ,\ t_i\longleftarrow 0$\;

\For{$t = 1,\ldots,m$}{
	\While{
	$\max_{k\in[n], s\leq t}{\bigg\lbrace v_k(s,t) - 	2\Big( v_k(s_k,t_k) + V_{-k}(s,t) \Big) \bigg\rbrace} \geq 0$ \;} {
		$k',s' \longleftarrow$ arguments maximizing the expression \;
		$s_{k'} \longleftarrow s'\ ,\ t_{k'} \longleftarrow t$ \; 
		$(s_i,t_i) \longleftarrow (0,0)$ for all $i$ with $s_i\geq s'$ \;
		$t_i \longleftarrow s'-1$ for $i$ with $s_i < s' \leq t_i$ \; 
	}

}
\Return{$\big\lbrace (s_i,t_i) \big\rbrace_{i\in[n]}$}
}
\caption{Discrete Utilitarian Welfare Approximation}
\end{algorithm}

Our algorithm for the discretized instance works iteratively, where in the $t$-th iteration it finds a good division for the first $t$ items. We begin with the trivial null allocation of 0 items. Assuming that we have a good allocation for the first $t-1$ items, and for all $s\leq t$ and $k\in[n]$, we consider the \emph{cost} of giving items $s$ through $t$ to player $k$. This cost is comprised of two components. The first component is the value of a piece $(s_k,t_k)$ that player $k$ may currently own, and has to give up in order to get the new piece $(s,t)$. The second component is the sum of values that the other players to which the items $s$ through $t$ are assigned obtain from these items. We denote this second component by $V_{-k}(s,t)$. We only give the segment $(s,t)$ to player $k$ if her total value $v_k(s,t)$ for this segment is at least twice the cost of giving her this segment. We continue trying to find a player $k'$ and a segment $(s',t)$ ending at item $t$ whose value exceeds twice the cost, and make changes until there are no such player and segment, at which point we move on to the next item $t+1$.

Observe that in the algorithm, each interval $(s,t)$ can be given to player $i$ at most once; this immediately implies that the running time of the algorithm is polynomial in the number of players $n$ and the number of items $m$.
For analyzing the approximation ratio of the algorithm, we use indicator variables $x_i^j$, for $i\in[n]$ and $j\in[m]$. At each step in the algorithm, we will have $x_i^j=1$ if and only if player $i$ owned the item $j$ at some point until now. 

\begin{lemma}\label{lem:cuo}
At any iteration $t$ of the above algorithm, we have
\begin{equation*}
\sum_{i\in[n]}{v_i(s_i,t_i)} \leq \sum_{i\in[n]}{\sum_{j\in[m]}{x_i^j\cdot v_i(j)}} \leq 2\cdot\sum_{i\in[n]}{v_i(s_i,t_i)} 
\end{equation*}
(where the values are as in the end of the $t$-th iteration).
\end{lemma}

\begin{proof}
The first inequality trivially holds, and we prove the second by induction on $t$. The second inequality clearly holds at the beginning of the step $t=1$; we show that if it holds at the beginning of some step $t$, then it must still hold at the end of this step. At the beginning of the $t$-th step, item $t$ is unallocated. If the while loop was not executed even once in this iteration, none of the expressions $\sum_{i\in[n]}{\sum_{j\in[m]}{x^i_j\cdot v_i(j)}}$ and $\sum_{i\in[n]}{v_i(s_i,t_i)}$ has changed, and the claim still holds. Otherwise, consider some iteration of the while loop. In such an iteration, the increase in $\sum_{i\in[n]}{\sum_{j\in[m]}{x_i^j\cdot v_i(j)}}$ is upper-bounded by $v_{k'}(s',t)$. The expression $\sum_{i\in[n]}{v_i(s_i,t_i)}$ also gains $v_{k'}(s',t)$, but in addition loses $v_k(s_k,t_k) + V_{-k}(s,t)$; however, from the while loop condition we have that $v_{k'}(s',t) - \big( v_k(s_k,t_k) + V_{-k}(s,t) \big) \geq \frac{1}{2}\cdot v_{k'}(s',t)$. Therefore, the increase to the right-hand side of the inequality is at least as large as that of the left-hand side, and the inequality is maintained. Since this holds for every iteration of the while loop, this still holds at the end of step $t$, as required.
\end{proof}

\begin{theorem}
Algorithm 2 returns an $8$-approximation of the discrete utilitarian optimum. 
\end{theorem}

\begin{proof}
Fix an instance, let $\big\lbrace (s_i^A, t_i^A) \big\rbrace_{i\in[n]}$ be the algorithm's final output on this instance, and let $\big\lbrace (s_i^*, t_i^*) \big\rbrace_{i\in[n]}$ be the optimal division, with total utility $OPT = \sum_{i\in[n]}{v_i(s_i^*,t_i^*)}$.

For every player $k$, consider the iteration $t_k^*$, in which the rightmost item given to $k$ in the optimal division was first considered. Let $(s_k',t_k')$ be the segment given to player $k$ at the end of this iteration. When iteration $t_k^*$ ends, it has to be that $v_k(s_k^*,t_k^*) \leq 2\big( v_k(s_k',t_k') + V_{-k}(s_k^*,t_k^*) \big)$ (where $V_{-k}(s_k^*,t_k^*)$ is with respect to the division set by the algorithm at this point). Note that $v_k(s_k',t_k') = \sum_{j=s_k'}^{t_k'}{x^k_j\cdot v_k(j)}$ and that $V_{-k}(s_k^*,t_k^*) \leq \sum_{j=s_k^*}^{t_k^*}{\sum_{i\neq k}{x^i_j\cdot v_i(j)}}$. Combining all this, we get
\begin{align*}
OPT  = \sum_{k\in[n]}{v_k(s_k^*,t_k^*)} & \leq \sum_{k\in[n]}{2\cdot \Big( \sum_{j=s_k'}^{t_k'}{x_j^k\cdot v_k(j)} + \sum_{j=s_k^*}^{t_k^*}{\sum_{i\neq k}{x^i_j\cdot v_i(j)}} \Big)} \\
	& = 2\cdot \Big( \sum_{k\in[n]}{\sum_{j=s_k'}^{t_k'}{x_j^k\cdot v_k(j)}} + \sum_{k\in[n]}{\sum_{j=s_k^*}^{t_k^*}{\sum_{i\neq k}{x^i_j\cdot v_i(j)}}} \Big) \\
	& \leq 2\cdot \Big( \sum_{k\in[n]}{\sum_{j\in[m]}{x_j^k\cdot v_k(j)}} + \sum_{k\in[n]}{\sum_{j\in[m]}{x_j^k\cdot v_k(j)}} \Big) \\
	& = 4\cdot \sum_{k\in[n]}{\sum_{j\in[m]}{x_j^k\cdot v_k(j)}} \leq 8\cdot \sum_{i\in[n]}{v_i(s_i^A,t_i^A)} \;.
\end{align*}
The second inequality holds since for every $k\neq k'$ the segments $(s_k^*,t_k^*)$ and $(s_{k'}^*,t_{k'}^*)$ are disjoint, as $\big\lbrace (s_i^*, t_i^*) \big\rbrace_{i\in[n]}$ is a division. The last inequality follows from Lemma~\ref{lem:cuo}.
\end{proof}

Combining Algorithm~\ref{alg:discretize} and Algorithm~\ref{alg:cuo} we get:


\begin{corollary}
For every $\epsilon > 0$, it is possible to find a division achieving utilitarian welfare within $8\big( 1+(n-1)\epsilon \big)$ of the optimum in time polynomial in $n$ and $1/\epsilon$.
\end{corollary}

\subsection{Fixed-Parameter Tractable Approximations}

Suppose that we have a relatively small number of players $n$, but that the social efficiency of the division is of much importance. We show that divisions that are within a factor of $1+\epsilon$ of the social optimum (for both utilitarian and egalitarian welfare) can be computed in time exponential in the number of players, but polynomial in  $\frac{1}{\epsilon}$.\footnote{Recall that we assume the oracle model; if the valuation functions are given explicitly, we also have polynomial dependence on the size of the input.} Using the terminology of the theory of Parametrized Complexity~\cite{DF99} we say that these approximations are \emph{fixed-parameter tractable} with respect to the number of players $n$. 

\begin{theorem}\label{thm:fpt-util}
For every $\epsilon > 0$, it is possible to find a division achieving utilitarian welfare within $1+\epsilon$ of the optimum in time  $2^n\cdot poly(n,\frac{1}{\epsilon})$.
\end{theorem}

\begin{proof}
We begin by transforming the cake into a set $A = \{a_1,\ldots,a_m\}$ of discrete items, by running the discretization procedure (Algorithm~\ref{alg:discretize}). We then use dynamic programming to find the utilitarian optimal division for the discretized instance. We create a table $U$ having a row of length $m$ for each pair $(S,k)$, where $S\subseteq [n]$ is a non-empty subset of players, and $k\in S$. Denote by $u_j^{(S,k)}$ the value in the $j$-th column of row $(S,k)$ in $U$. The value $u_j^{(S,k)}$ will be the largest total utility obtainable by dividing the first $j$ items between the players of $S$, where the last item $j$ is given to player $k$. Therefore, the maximum utilitarian welfare in the discretized instance is $\max_{S,k}{\big\lbrace u_m^{(S,k)} \big\rbrace }$.

It therefore remains to show how the values in $U$ can be computed. We begin with the first column. For every $k\in[n]$, we have $u_1^{(\{k\},k)} = v_k(a_1)$; for any $S\neq \{k\}$ the pair $(S,k)$ is invalid and we set $u_1^{(S,k)} = -\infty$. Now, suppose that we have filled in the values in the first $j$ columns, and consider $u_{j+1}^{(S,k)}$. If $k\notin S$, again $(S,k)$ is invalid, and we set $u_{j+1}^{(S,k)} = -\infty$. If $k\in S$, then the maximum value is obtained by either extending $k$'s piece in the best division of the items $1$ through $j$ in which $k$ is gets item $j$, or by taking the best division of items $1$ through $j$ between the players $S\setminus \{k\}$, and giving item $j+1$ to $k$. Formally, 
\begin{equation*}
u_{j+1}^{(S,k)} = \max\bigg\lbrace u_j^{(S,k)} + v_k(a_{j+1}) \ ,\  \max_{i\in S\setminus \{k\}}\Big\lbrace u_j^{(S\setminus\{k\},i)}\Big\rbrace + v_k(a_{j+1}) \bigg\rbrace \;.
\end{equation*}
We get that the value of each table entry can be computed in time $O(n)$, and therefore filling the entire table $U$ and finding the utilitarian optimum can be done in time $2^n\cdot poly\big( n, \frac{1}{\epsilon} \big)$, as stated. Once we have the approximated welfare, the actual division can be easily computed by backtracking.
\end{proof}


\begin{theorem}\label{thm:fpt-egal}
For every $\epsilon > 0$, it is possible to find a division achieving egalitarian welfare within $1+\epsilon$ of the optimum in time  $2^n\cdot n\cdot \log_2\left(\frac{n}{\epsilon}\right)$.
\end{theorem}

\begin{proof}
We first show an \emph{additive} $\epsilon$-approximation to the egalitarian welfare 
that works as follows. Let $B\in [0,1]$ be our goal egalitarian welfare. We construct a vector $C$ of cut positions, with one entry for every non-empty subset $S\subseteq [n]$ of players; the value $C(S)$ will be the leftmost point $a\in[0,1]$ such that it is possible to divide the interval $[0,a]$ between the players of $S$ and have each $i\in S$ have utility $B$. Clearly, the egalitarian welfare $B$ is feasible if and only if $C\big([n]\big)\leq 1$.

We now show how to compute the values of $C$. First, for each $i\in[n]$, we can set $C(\{i\}) = v_i^{-1}(0,B)$. Suppose that for $k<n$ we have filled all the entries $C(S)$ for $|S|\leq k$. Let $S$ be such that $|S| = k+1$; then $C(S) = \min_{i\in S}{v_i^{-1}\big(C(S\setminus\{i\}),B\big)}$ (if there is no such value, we can mark $C(S) = \infty$). Therefore, each entry can be computed in time $O(n)$; by using binary search, we can find a value $B$ within an additive distance of $\epsilon$ of the egalitarian optimum in $\log_2{\left(\frac{1}{\epsilon}\right)}$ iterations. 

Now, note that the optimal egalitarian welfare is at least $\frac{1}{n}$, since a proportional division (which has at least this welfare) is guaranteed to exist. Therefore, by continuing the binary search for $\log_2{n}$ more iterations, we can find a value of $B$ which is within a $(1-\epsilon)$-\emph{multiplicative} factor of the optimum, and the stated complexity follows. As in the previous algorithm, the actual division can be easily computed by backtracking.
\end{proof}

\section{Hardness}

We show that all of the four problems defined in Section~\ref{sec:prelim} are NP-complete in the strong sense. Note that membership in NP is straightforward, as a division achieving the required welfare can serve as a witness for that instance; we thus concentrate on proving hardness. 

\subsection{Hardness of Maximizing Egalitarian Welfare}

We prove that \egal~is strongly NP-complete and hard to approximate to a factor of $2-\epsilon$ for any $\epsilon>0$. We show this using a reduction from the classic problem of \DM~\cite{GJ79}. In this problem, one is given three sets $X,Y,Z$ of cardinality $n$ each, as well as a set $E \subseteq X\times Y\times Z$, and needs to decide whether there exists a subset $E'\subseteq E$ of cardinality $n$ that covers $X, Y$ and $Z$.

Our reduction borrows its main ideas from the proofs of Bez{\'a}kov{\'a} and Dani~\cite{BD05} for non-connected divisions in the discrete setting, which itself uses ideas from Lenstra, Shmoys and Tardos~\cite{LST90}. 
However, the adjustment to the continuous setting with connected divisions is somewhat intricate and needs to be done carefully.

\begin{theorem}\label{thm:ceo-hard}
\egal~and \egalD~are NP-complete in the strong sense. Furthermore, for every $\epsilon > 0$ there is no $(2-\epsilon)$ approximation for either of the problems, unless P=NP. 

This holds even if the valuation functions of the players are piecewise-uniform, and are given explicitly to the algorithm.
\end{theorem}

\begin{proof}
We show a polynomial-time reduction from \DM~to \egal. Let $X,Y,Z$ and $E\subseteq X\times Y\times Z$ be an input to \DM. We construct a set of piecewise-constant valuations and a bound $B$ as an input for \egal; this instance can be transformed into an equivalent one with piecewise-uniform valuations.

For convenience, we take the cake to be the interval $\big[0,2|E|\big]$ rather than $[0,1]$. We will think of the cake as being sectioned into $|E|$ ``sections'' of length $2$, where the second unit of each section is used for separation from the next section.\footnote{Indeed, the last section needs not have such a unit; however, we leave it there in order to treat it identically to all the other sections.}
The set of players we create has players of three types: ``triplet players'', ``ground sets players'' and ``separation players''. In what follows we describe the valuation functions of all the players, by their type; for the bound, we set $B = \frac{1}{|E|}$.

\begin{itemize}
\item \textbf{Triplet Players:} We create a player for every $z\in Z$. For each $e_i\in E$ such that $z$ appears in the triplet $e_i$, the player created for $z$ has value of $\frac{1}{2|E|}$ for each of the intervals $\big(2(i-1), 2(i-1) + \frac{1}{4} \big)$ and $\big( 2(i-1) + \frac{3}{4}, 2(i-1) + 1 \big)$ in the left half of the $i$-th section.

Denote by $m_z$ the number of such triplets $e_i$ in $E$. To keep the value of the entire cake at $1$ for each player, we will divide the missing value $1 - \frac{m_z}{|E|}$ between the right halves of all sections. Specifically, player $z$ will additionally have value $\frac{|E|-m_z}{2|E|^2}$ for every interval $\big(2(j-1)+\frac{6}{5}, 2(j-1) + \frac{7}{5} \big)$ and $\big( 2(j-1) + \frac{8}{5}, 2(j-1) + \frac{9}{5} \big)$, for all $1\leq j\leq |E|$.

\item \textbf{Ground Sets Players:} For $x\in X$, let $m_x$ be the number of triplets in $E$ in which $x$ appears. We create $m_x-1$ identical players for $x$. For every $e_i\in |E|$ such that $x$ appears in $e_i$, all of $x$'s players will have valuation of $\frac{1}{|E|}$ for the interval $\big(2(i-1)+\frac{1}{4},2(i-1)+\frac{1}{2}\big)$ in the left half of the $i$-th section. Again, in order to complement these valuations to $1$, they will also assign a value of $\frac{|E|-m_x}{2|E|^2}$ for each of the intervals $\big(2(j-1)+\frac{6}{5}, 2(j-1) + \frac{7}{5} \big)$ and $\big( 2(j-1) + \frac{8}{5}, 2(j-1) + \frac{9}{5} \big)$, for all $1\leq j\leq |E|$.

We similarly create $m_y-1$ identical players for every $y\in Y$. For each $e_i\in E$ in which $y$ appears we have these players give value of $\frac{1}{|E|}$ to the interval $\big(2(i-1)+\frac{1}{2},2(i-1)+\frac{3}{4}\big)$, and complement this by giving value of $\frac{|E|-m_y}{2|E|^2}$ to each of the intervals $\big(2(j-1)+\frac{6}{5}, 2(j-1) + \frac{7}{5} \big)$ and $\big( 2(j-1) + \frac{8}{5}, 2(j-1) + \frac{9}{5} \big)$, for all $1\leq j\leq |E|$.

\item \textbf{Separation Players:} We finally create $3|E|$ separation players. For every segment $1\leq i\leq |E|$ we have a player $s_{3i-2}$ have valuation of $1$ for the interval $\big( 2(j-1)+1,3(j-1)+\frac{6}{5} \big)$, another player $s_{3i-1}$ have valuation $1$ for $\big( 2(i-1)+\frac{7}{5},2(i-1)+\frac{8}{5}\big)$, and a third player $s_{3i}$ have valuation $1$ for $\big( 2(i-1)+\frac{9}{5},2(j-1) + 2 \big)$.
\end{itemize}

Figure~\ref{fig:ceo-hard} illustrates the structure of the preferences in one segment. In this example, we consider some triplet $e_i = (x_j,y_k,z_\ell)\in E$, and show the section of the cake created for it, with the preferences of the players who desire some piece of it.

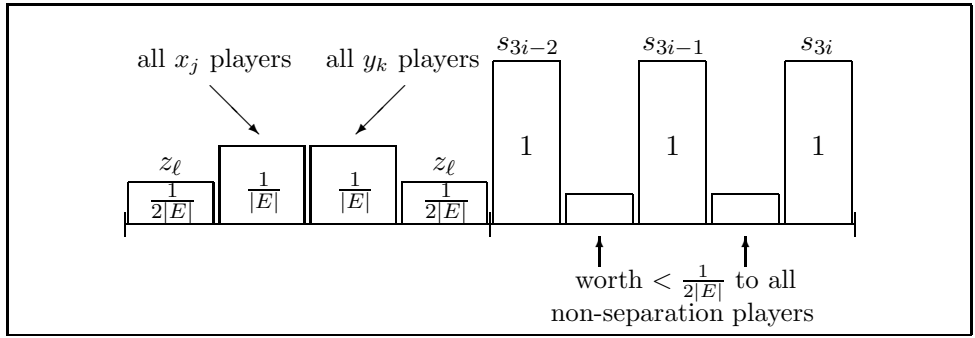
\begin{figure}[bth]
\begin{center}
\framebox[1.1\width]{

\setlength{\unitlength}{8mm}
\begin{picture}(14.1,5.2)

\put(0.9,1.7){\line(1,0){12}}
\put(0.9,1.5){\line(0,1){0.4}}
\put(6.9,1.5){\line(0,1){0.4}}
\put(12.9,1.5){\line(0,1){0.4}}

\multiput(0.95,1.7)(4.5,0){2}{\line(0,1){0.7}}
\multiput(2.35,1.7)(4.5,0){2}{\line(0,1){0.7}}
\multiput(0.95,2.4)(4.5,0){2}{\line(1,0){1.4}}
\multiput(1.28,1.97)(4.5,0){2}{$\frac{1}{2|E|}$} 
\put(1.45,2.6){$z_\ell$} \put(5.95,2.6){$z_\ell$}  

\multiput(2.45,1.7)(1.5,0){2}{\line(0,1){1.3}} 
\multiput(3.85,1.7)(1.5,0){2}{\line(0,1){1.3}} 
\multiput(2.45,3.0)(1.5,0){2}{\line(1,0){1.4}}
\multiput(2.88,2.15)(1.5,0){2}{$\frac{1}{|E|}$} 
\put(1.1,4.3){\small{all $x_j$ players}} 
\put(4.2,4.3){\small{all $y_k$ players}} 
\put(2.3,4){\vector(1,-1){0.8}}
\put(5.5,4){\vector(-1,-1){0.8}}


\multiput(6.95,1.7)(2.4,0){3}{\line(0,1){2.7}} 
\multiput(8.05,1.7)(2.4,0){3}{\line(0,1){2.7}} 
\multiput(6.95,4.4)(2.4,0){3}{\line(1,0){1.1}} 
\multiput(7.38,2.85)(2.4,0){3}{$1$}
\put(7.0,4.6){$s_{3i-2}$} \put(9.4,4.6){$s_{3i-1}$} \put(12.0,4.6){$s_{3i}$}

\multiput(8.15,1.7)(2.4,0){2}{\line(0,1){0.5}} 
\multiput(9.25,1.7)(2.4,0){2}{\line(0,1){0.5}} 
\multiput(8.15,2.2)(2.4,0){2}{\line(1,0){1.1}} 
\multiput(8.7,1.0)(2.4,0){2}{\vector(0,1){0.5}} \put(8.3,0.65){\small{worth $<\frac{1}{2|E|}$ to all}}
\put(7.9,0.1){\small{non-separation players}}

\end{picture}}
\end{center}\caption{The valuations of the players for the section created for $e_i = (x_j,y_k,z_\ell)\in E$. Note that there are $m_{x_j}-1$ identical players for $x_j$ and $m_{y_k}-1$ identical players for $y_k$.}
\label{fig:ceo-hard}
\end{figure}

It is straightforward to observe that the construction above can be carried out in polynomial time. Also, all the numbers created in this instance can be represented by a number of bits logarithmic in the input size.


Suppose that $(X,Y,Z,E)\in$\DM, and let 
$E'\subseteq E$ be a cover of $X\cup Y\cup Z$. For each $z\in Z$ there is a unique $e_i\in E'$ in which $z$ appears. Give the (unique) triplet player $z$ the left half of the $i$-th section $\big( 2(i-1), 2(i-1) + 1 \big)$. Next, consider the ground sets players. For every $x\in X$, there is a unique $e_i\in E'$ in which $x$ appears; hence, there are $m_x-1$ triplets $e_j \in E\setminus E'$ in which $x$ appears. Since no piece of the corresponding sections was given so far, we can give an interval of the form $\big( 2(j-1), 2(j-1) + \frac{1}{2}\big)$ to each of $x$'s players. Similarly, we have $m_y-1$ sections corresponding to $e_j\in E\setminus E'$ for $e_j$'s that contain $y$ in which the interval $\big( 2(j-1) + \frac{1}{2}, 2(j-1) + 1\big)$ is still available (and is worth $\frac{1}{|E|}$ to all of $y$'s players). This satisfies all the ground-sets players as well. Note that so far we have only given out the left half of each section; therefore, we can give each separation player her entire (single and unique) desired piece, dividing the parts between them arbitrarily. This gives each of them utility of $1$. We obtain that the division above indeed provides each player a piece of value at least $B = \frac{1}{|E|}$.


Conversely, suppose that $(X,Y,Z,E)\notin$\DM; we wish to prove that in this case, no division has egalitarian welfare $>\frac{1}{2|E|}$. Since $(X,Y,Z,E)\notin$\DM, it has to be that in every $E'\subseteq E$ of cardinality $n$ that covers all of $Z$, there is either some $x\in X$ or some $y\in Y$ that appears twice. Consider the egalitarian welfare maximizing division. Note that an interval containing a desired piece of a separation player could not be given to any other player, as this would starve that separation player completely, and the egalitarian welfare will be $0$. 
Hence, if any non-separation player gets a piece which intersects the right half of some section, this piece must be of value strictly smaller than $\frac{1}{2|E|}$ (or it would contain the entire desired piece of some separation player), and we are done. If this is not the case, then each non-separation player must get her piece from the left half of some section. Assume, towards contradiction, that each such player is given a piece of value at least $\frac{1}{2|E|} + \epsilon$ for some $\epsilon > 0$. This means that each $z$ player gets a piece that contains the interval $\big( 2(i-1) + \frac{1}{4}, 2(i-1) + \frac{3}{4} \big)$, and in particular consumes the desired pieces of all the ground players in this half-section. Let $E'$ be the set of $e_i$'s such that some triplet player has received her piece from the $i$-th section. This is a subset of $E$ of size $n$ that covers $Z$, and hence contains some $x\in X$ or $y\in Y$ twice; suppose w.l.o.g.~that it is $x$. We have $m_x - 1$ ground set players created for $x$, all interested in a total of $m_x$ intervals. However, at least two of these intervals are fully consumed by triplet players. Thus, at least two of $x$'s players must share a desired interval, and hence at least one of these players gets a piece of value at most $\frac{1}{2|E|}$; a contradiction. We conclude that if $(X,Y,Z,E)\notin$\DM, no partition can provide all the players with a piece of value more than $\frac{1}{2|E|}$.

Note that in the instance we created, every player has the same value density for all of her desired intervals intersecting with those of other players.
Therefore, for each player $i$, we can rescale the physical size of her other desired intervals (i.e.~those that only player $i$ is interested in), so that $\nu_i$ has a unique value whenever it is not zero, while maintaining the total value of every desired interval. Note that this does not change anything (except the cut locations) in the proof above, and therefore the proof holds for piecewise-uniform valuations as well.

The proof for \egalD~is analogous, and can easily be obtained by a straightforward partitioning of the cake created in the reduction into discrete indivisible chunks. 
\end{proof}



To prove the hardness of maximizing utilitarian welfare, we show a reduction from \egalD. However, it is convenient to show this reduction in two steps: we first define a new problem, to which \egalD~can be easily reduced, and then reduce this problem to \util~and \utilD. Our ``auxiliary problem'' is the following:

\problem{\mcspLong~(\mcsp)}
{An integer $m$, and sets $A_1,\ldots,A_n$ of segments of $[m]$.}
{Does there exist a set $S\subseteq \bigcup_{i\in[n]}{A_i}$ of $n$ \emph{disjoint} segments having $|S\cap A_i| = 1$ for all $i$?}

Given an instance of \egalD, we can define each $A_i$ as the set of consecutive subsequences of $A = [m]$ (the sequence of items) with value at least the threshold $B$ to $i$. This yields (in polynomial time) an instance of \mcsp~for which the answer is ``yes'' if and only if the original instance (for \egalD) was a ``yes'' instance. In this sense, \mcsp~can be thought of as a ``cleaner'' generalization of \egalD, which abstracts away the item values, and just divides all possible segments to ``acceptable'' and ``non-acceptable'' ones. 

By the observation above, we immediately get that \mcsp~is NP-complete. We now show that maximizing utilitarian welfare is hard using a reduction from \mcsp.

\begin{theorem}\label{thm:cuo-hard}
\util~and \utilD~are NP-complete in the strong sense.

This holds even if the valuation functions of the players are piecewise-uniform, and are given explicitly to the algorithm.
\end{theorem}

\begin{proof}
We show a polynomial-time transformation from \mcsp~to \util; as in the egalitarian case, the proof can be easily modified to work for the discrete version \utilD~as well. Given an integer $m$ and sets $A_1,\ldots,A_n$ of segments of $[m]$, we create an instance of \util; again, we first create an instance with piecewise-uniform valuations.

For convenience, we take the cake to be the interval $\big[0,m + 2\cdot\sum_{i\in[n]}{|A_i|}-2n \big]$. We think of the cake as being composed of two parts: the part $[0,m]$ is the ``segments range'', and the remainder is the ``compensation range'', which itself will divided into different regions, one for each set of segments $A_i$. We will set the bound to $B = \frac{4}{3}\sum_{i\in[n]}{|A_i|} - n$. It remains to create the players and their preferences.

For each set $A_i$ of segments we create ``segment players'' $p_i^1,\ldots,p_i^{|A_i|}$ and ``separation players'' $q_i^1,\ldots,q_i^{|A_i|-1}$. 
We denote $C_i = m + 2\cdot\sum_{i'<i}{(|A_{i'}|-1)}$; the players created for $A_i$ will only have non-zero values in the interval $[0,m]$ (the ``segments range'') and in the interval $[C_i,C_i + 2(|A_i| -1)]$ ($A_i$'s ``compensation range''). We now describe these preferences in detail.

\begin{itemize}
\item \textbf{$A_i$'s segment players:} Let $s_i^1,\ldots,s_i^{|A_i|}$ be the segments of $A_i$, and denote by $b(s_i^j)$ the beginning index of $s_i^j$, and by $e(s_i^j)$ its ending index. 

Every player $p_i^j$ for some $1\leq j \leq |A_i|$ has value $\frac{1}{3}$ for the interval $\big( b(s_i^j)-1,e(s_i^j) \big)$. In addition, each such player will have non-zero value for two more intervals in the ``compensation range'': Player $p_i^1$ will have value $\frac{1}{3}$ for each of the intervals $(C_i+1,C_i+2)$ and $(C_i+3,C_i+4)$. Player $j$, for $2\leq j\leq |A_i|-1$, will have value $\frac{1}{3}$ for each of the intervals $(C_i+2C_i-3,S+2j-2)$ and $(C_i+2j-1,C_i+2j)$. Finally, player $p_i^{|A_i|}$ will have value $\frac{1}{3}$ for each of the intervals $(C_i+2|A_i|-5,C_i+2|A_i|-4)$ and $(C_i+2|A_i|-3,C_i+2|A_i|-2)$.
These players will have zero value for the rest of the cake.

\item \textbf{$A_i$'s separation players:} Each of these players has non-zero value for only one interval in $A_i$'s compensation range, and no value for any other part of the cake. Specifically, for every $1\leq j\leq |A_i|-1$, player $q_i^j$ has value $1$ for the interval $(C_i+2j-2,C_i+2j-1)$.	
\end{itemize}

Figure~\ref{fig:cuo-hard} shows the structure of the players of one set $A_i$ (in this example, having $|A_i| = 7$) in the compensation range.

\begin{figure}[bth]
\begin{center}
\framebox[1.1\width]{

\setlength{\unitlength}{8mm}
\begin{picture}(14.1,5.2)

\put(0.9,1.7){\line(1,0){12}}
\put(0.9,1.5){\line(0,1){0.4}}
\put(12.9,1.5){\line(0,1){0.4}}

\multiput(0.95,1.7)(2,0){6}{\line(0,1){2.1}} 
\multiput(1.85,1.7)(2,0){6}{\line(0,1){2.1}} 
\multiput(0.95,3.8)(2,0){6}{\line(1,0){0.9}} 
\multiput(1.3,2.6)(2,0){6}{$1$} 
\put(1.25,4.0){$q_i^1$} \put(3.25,4.0){$q_i^2$} \put(5.25,4.0){$q_i^3$} \put(7.25,4.0){$q_i^4$} \put(9.25,4.0){$q_i^5$} \put(11.25,4.0){$q_i^6$}

\multiput(1.95,1.7)(2,0){6}{\line(0,1){0.7}} 
\multiput(2.85,1.7)(2,0){6}{\line(0,1){0.7}} 
\multiput(1.95,2.4)(2,0){6}{\line(1,0){0.9}} 
\multiput(2.03,1.9)(2,0){6}{$1/3$} 

\put(2.25,0.4){$p_i^1$} \put(3.25,0.4){$p_i^2$} \put(5.25,0.4){$p_i^3$} \put(7.25,0.4){$p_i^4$} \put(9.25,0.4){$p_i^5$} \put(11.25,0.4){$p_i^6$} \put(12.25,0.4){$p_i^7$} 

\put(2.4,0.95){\vector(0,1){0.5}} \put(2.6,0.95){\vector(3,1){1.5}}
\multiput(3.6,0.95)(2,0){5}{\vector(3,2){0.75}} \multiput(3.3,0.95)(2,0){5}{\vector(-3,2){0.75}} 
\put(12.3,0.95){\vector(-3,1){1.5}} \put(12.5,0.95){\vector(0,1){0.5}}

\end{picture}}
\end{center}\caption{The valuations of the players of one set $A_i$ in the compensation range. Each segment player $p_i^j$ has an additional piece of value $1/3$ in the segments part (not shown here).}
\label{fig:cuo-hard}
\end{figure}
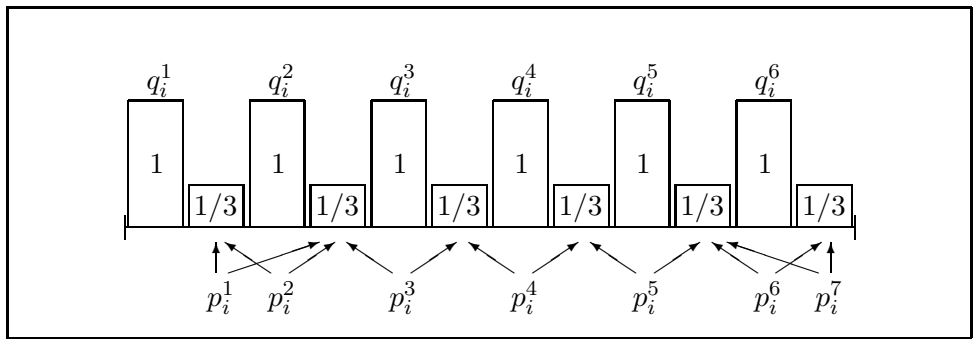

Clearly, the transformation can be computed in polynomial time; also note that the numbers created in this instance are all bounded by a polynomial of $m$ and $n$. 

It thus remains to show that a division of utilitarian welfare at least $B$ exists in such an instance if and only if we can choose a set of $n$ disjoint segments from $\bigcup_{i\in[n]}{A_i}$ such that exactly one segment is chosen from each $A_i$.

Suppose first that we can indeed choose $n$ such segments; we now describe a division of the cake. First, give each separation player her (single and unique) desired interval (in whole). This contributes $\sum_{i\in[n]}{|A_i|}-n$ to the welfare. Now, for every $A_i$, let $s_i^k$ be the segment chosen for the packing. We give player $p_i^k$ the interval $\big( b(s_i^k)-1,e(s_i^k) \big)$; this is possible, since all these segments are disjoint. To each player $p_i^j$ with $j<k$ we give the interval $(C_i + 2j-1,C_i + 2j)$, and to each player $p_i^j$ with $j>k$ we give the interval $(C_i + 2j - 3,C_i + 2j - 2)$. This gives every segment player utility of $\frac{1}{3}$, and contributes the missing amount of $\frac{1}{3}\sum_{i\in[n]}{|A_i|}$ to the utilitarian welfare.

Suppose now that there exists a division with utilitarian welfare of $B$. 
We complete the proof by showing that for every $A_i$ the maximum contribution of the players of $A_i$ to the utilitarian welfare is $\frac{4}{3}\cdot|A_i|-1$, and that this contribution is achievable if and only if at least one of the segment players of $A_i$ gets all her complete desired interval from the segments range.

Fix some $i$, and consider the players created for the set $A_i$. Note first that to maximize the sum of utilities of these players, no segment player can have utility exceeding $\frac{1}{3}$. This holds because any piece worth more to such a player will necessarily contain the entire desired piece of some separation player. If we give this piece to the segment player, the total utility of the two players is bounded by $1$; however, if we give the segment player a smaller piece and let the separation player have her entire desired piece we achieve total utility $>1$ without affecting any of the other players. Therefore, we get that the only way to provide total utility of at least  $\frac{4}{3}|A_i|-1$ to the players created for $A_i$ is to give each separation player her entire desired piece, and let each segment player have a piece she values exactly at $\frac{1}{3}$. However, it is easy to observe that to do that, at least one of the segment players must get her entire desired piece from the segments range.

As in the proof for \egal, we would now like to transform the valuations into piecewise-uniform ones. However, note that the property that each player has uniform value density on all the intervals desired also by other players does not hold in the current construction.
Namely, a segment player has a desired interval common with other players in the compensation range, and in addition another desired interval, which may have different value density in the segments range, and this interval may also intersect other players' desired intervals. 
However, this can be remedied by splitting the desired intervals in the compensation range. Let $(a,b)$ be one of the intervals in the compensation range valued as worth $\frac{1}{3}$ by two players $p_i^j$ and $p_i^{j+1}$ (a similar idea can be used for the intervals desired by three players). We change the valuations of the players so that now player $p_i^j$ values each of the intervals $\big( a,\frac{a+b}{4} \big)$ and $\big(\frac{2(a+b)}{4}, \frac{3(a+b)}{4} \big)$ as worth $\frac{1}{6}$ and similarly $p_i^{j+1}$ values each of the intervals $\big( \frac{a+b}{4} , \frac{2(a+b)}{4} \big)$ and $\big( \frac{3(a+b)}{4} , b \big)$ as worth $\frac{1}{6}$. After this change, each player has at most one desired interval (in the segments range) that intersects other players' intervals, and we can 
transform the valuations into piecewise-uniform ones, as in the prveious proof.
This also maintains the property that the players created for $A_i$ can get total utility of $\frac{4}{3}|A_i| - 1$ iff at least one segment player gets her entire desired piece from the segments range.
\end{proof}

The strong NP-hardness of \util~and \utilD~implies the following corollary:

\begin{corollary}
There is no FPTAS for either \util~nor \utilD.
\end{corollary}

\section{Welfare Maximization with Non-Connected pieces}

In this section we analyze the problem of welfare maximization when each player may get a \emph{collection} of intervals.
We first show that if the valuation functions are piecewise-constant and are given explicitly to the algorithm, the problem can be easily solved in polynomial time using a linear program almost identical to the one used by Cohler et al.~\cite{CLPP11}. 

\begin{theorem}\label{thm:non-conn-lp}
Given a set of $n$ piecewise-constant valuation functions (i.e.~for each $i\in[n]$ the list of intervals in which the value density function attains different values, along with the value for each such interval), it is possible to find a division maximizing the utilitarian (resp.~egalitarian) welfare in polynomial time.
\end{theorem}

\begin{proof}
Given the set of valuation functions, it is useful to divide the interval $[0,1]$ into a set of subintervals on which the value density function $\nu_i$ of \emph{every} player $i$ is constant. This can be done by simply taking the union of all the boundary points of all the intervals in the description of each $\nu_i$, together with the points $0$ and $1$, and considering the set of intervals formed be every two consecutive points in this set. Denote this set of intervals by $\mathcal{J}$. Let $I\in\mathcal{J}$ and $i\in[n]$; since all $a\in I$ have the same value $\nu_i(a)$ we slightly abuse notation and refer to this value as $\nu_i(I)$.

Clearly, the division maximizing the utilitarian welfare gives each interval $I\in \mathcal{J}$ to the player $i$ with the maximum $\nu_i(I)$. Therefore, finding such a division is straightforward.

For finding a division maximizing the egalitarian welfare, we can solve the following linear program:
\begin{align}
 \text{maximize\ } t &&  & \notag \\
 \text{subject to} & & \sum_{I\in\mathcal{J}}{\nu_i(I)\cdot x_i^I} & \geq t && \forall i\in[n] \label{eq:pl-util} \\
  & & \sum_{i\in[n]}{x_i^I} & \leq 1  && \forall I\in\mathcal{J} \label{eq:interval} \\
 & & x_i^I &\geq 0 && \forall i\in[n], I\in\mathcal{J}  \notag
\end{align}
Here, $x_i^I$ indicates the portion of interval $I$ given to player $i$. The constraint~\eqref{eq:pl-util} makes sure that each player obtains utility of at least $t$ (which is the variable the program maximizes), and the constraint~\eqref{eq:interval} assures that we do not give away more than 100\% of any interval. Therefore, by solving the above program, we can find the egalitarian-optimal division in polynomial time.
\end{proof}

In contrast to this positive result, it turns out that maximizing welfare is \emph{impossible} if instead of receiving the explicit functions, we only get oracle access to the valuations. In particular, we show that no deterministic algorithm (even super-polynomial) can find a division approximating the utilitarian or egalitarian optimum by a factor smaller than $n$. Note that this bound is tight, as every \emph{proportional} division\footnote{Recall that a division is said to be proportional if it gives each player what she considers to be at least $1/n$ of the total value of the cake.} approximates utilitarian and egalitarian welfare by at least $n$, and many algorithms for finding proportional divisions do exist in the queries model (see, e.g.~\cite{RW98} for a survey).


\begin{theorem}\label{thm:non-conn-lb}
For any $\epsilon > 0$, no deterministic algorithm working in the oracle input model can approximate utilitarian or egalitarian welfare to a factor of $n-\epsilon$, when non-connected pieces are allowed.
\end{theorem}

\begin{proof}
We discuss utilitarian welfare; the arguments for egalitarian welfare are analoguous.
Let $A$ be a deterministic cake division algorithm working in the oracle input model, and fix some $n\in\mathbb{N}$ and $\epsilon>0$. Consider the operation of the algorithm on the set of preferences in which all players value the entire cake uniformly. In this case, the utilitarian welfare obtained cannot exceed $1$. We will now show that for any $\epsilon'>0$ we can construct a different set of preferences on which the algorithm outputs the same division (with the same welfare), but for which there exists a division achieving utilitarian welfare of $(1-\epsilon')n$. The theorem will follow by choosing $\epsilon' = \epsilon/n$.

Let $0 = p_0 < p_1 < \ldots < p_{k-1} < p_k = 1$ be the set of (distinct) points that appear in the operation of the algorithm on the input above. I.e.~$\{p_i\}_{i=0}^k$ is the set of all points $a,b$ for which the algorithm makes a query $v_i(a,b)$ or receives an answer $b = v_i^{-1}(a,x)$, and all the points $c$ in which the algorithm makes cuts in its output division. We create a new instance in which the valuations in the interval between two each consecutive such points $(p_j,p_{j+1})$ are ``rearranged''. The value of this interval in the original instance, as well as in the new instance, is $\ell_j = p_{j+1} - p_j$. We divide this interval into $n+1$ ``slivers'': the $i$-th sliver ($1\leq i\leq n$) is worth $\ell_j - \frac{\epsilon'}{k}$ to player $i$, and zero to everyone else. The $n+1$-st sliver of the interval is worth $\frac{\epsilon'}{k}$ for all the players. Formally, for each player $i$ and each $0\leq j\leq k$, set 
\begin{align*}
 v_i'\bigg( p_j + \frac{i-1}{n+1}\cdot\ell_j\ ,\ p_j + \frac{i}{n+1}\cdot\ell_j \bigg) & = v_i(p_j,p_{j+1}) - \frac{\epsilon'}{k} \\
 v_i'\bigg( p_j + \frac{n}{n+1}\cdot\ell_j\ ,\ p_{j+1} \bigg) & = \frac{\epsilon'}{k}
\end{align*}
and have player $i$ give value of $0$ to any other piece.

We now have that for every $0\leq j\leq j'\leq k$ it holds that $v'_i(p_j,p_{j'}) = v_i(p_j,p_{j'})$; furthermore, for every $0\leq j\leq k$ and every $x\in\mathbb{R}$ such that $A$ makes a query $v_i^{-1}(p_j,x)$ when executed on the preferences $\vec{v}$ we have ${v'}_i^{-1}(p_j,x) = v_i^{-1}(p_j,x)$. Therefore, consider the operation of the algorithm $A$ on the set of valuations $\vec{v'}$. The first query of $A$ on this instance is clearly identical to its first query on the instance $\vec{v}$, since before any queries are asked $A$ cannot distinguish between the two instances. However, as we observed, the answer to $A$'s first query with the new instance $\vec{v'}$ to the answer with the original instance $\vec{v}$. Since $A$ is deterministic, this implies that the next query asked by $A$ on $\vec{v'}$ is identical to that asked on $\vec{v}$. Continuing in this manner, we get that the entire operation of $A$ is identical on both instances. 

In particular, this implies that the cut points in the division produced for $\vec{v'}$ are all from the set $\{p_j\}_{j=0}^k$. However, this means that in this division, every player has the same value as it has in the division produced for $\vec{v}$, and therefore the utilitarian welfare is $1$. However, any division giving each player all ``slivers'' only she gives positive value to yields utilitarian welfare $> (1-\epsilon')n$, and the theorem follows.
\end{proof}

\section{Open Problems}

In this work we have taken the first steps in studying the problem of maximizing welfare in cake cutting with connected pieces. 
Many interesting problems related to this problem remain open. First and foremost, we believe that it should be possible to obtain 
a reasonable approximation for the problem of maximizing the egalitarian welfare. (We do have non-trivial algorithms that achieve linear-factor approximations, but we conjecture that better algorithms can be found.) 
We also conjecture that the approximation ratio for maximizing utilitarian welfare can be improved; it may also be interesting to see if other inapproximability results can be shown. 
Other interesting extensions include:

\begin{itemize}
\item \emph{Strategic Behavior:} One implicit assumption in our work was that we have access to the (true) valuations of the players. In reality, the players may have incentive to lie about their valuations.
Guo and Conitzer~\cite{GC10} and Han et al.~\cite{HSTZ11} have considered this problem for a somewhat different setting; the question of what can be achieved truthfully in our setting is still open.

\item \emph{2-Dimensional Cake:} The cake cutting literature has generally assumed a one-dimensional cake; indeed, for the purpose of maintaining fairness, which was its main focus, a 2-dimensional cake can be simply ``projected'' into one dimension, and divided fairly according to the projection. However, this may result in a significant loss of welfare. Therefore, maximizing welfare in allocation of 2-dimensional cakes may require completely different tools and techniques.
\end{itemize}

\paragraph{Acknowledgment.} We thank Zvi Gotthilf for many useful discussions.

\bibliographystyle{alpha}
\bibliography{CEO}

\end{document}